\documentclass[article,12pt]{article}
\usepackage{amsmath,amssymb,cite}
\usepackage{amsthm}
\usepackage{graphics,epsfig,color}
\usepackage[mathscr]{euscript}

\usepackage{comment}
\usepackage{enumerate}

\usepackage{graphicx}
\usepackage{graphics}
\usepackage{caption}
\usepackage{subfig}
\usepackage{hyperref}
\usepackage{tikz}
\usetikzlibrary{arrows}

\usepackage{tabularx}
\usepackage{bm}
\usepackage{multirow}
\usepackage{pifont}

\usepackage{bbold}

\theoremstyle{plain}
\newtheorem{theorem}{Theorem}[section]
\newtheorem{proposition}[theorem]{Proposition}
\newtheorem{lemma}[theorem]{Lemma}
\newtheorem{corollary}[theorem]{Corollary}

\theoremstyle{definition}
\newtheorem{definition}[theorem]{Definition}

\theoremstyle{remark}
\newtheorem{remark}[theorem]{Remark}
\newtheorem{example}[theorem]{Example}

\numberwithin{equation}{section}
\numberwithin{theorem}{section}

\newcommand{\wc}{\mathcal{W}}
\newcommand{\wcm}{\mathcal{W}_{\min}}
\renewcommand{\epsilon}{\varepsilon}
\renewcommand{\tilde}{\widetilde}
\renewcommand{\hat}{\widehat}

\newcommand{\xmark}{\ding{55}}%

\title{With a little help from my friends: essentiality vs opportunity in  group criticality\footnote{Michele Aleandri and Marco Dall'Aglio are member of the Gruppo Nazionale per l’Analisi Matematica, 
la Probabilità e le loro Applicazioni (GNAMPA) of the Istituto Nazionale di Alta Matematica (INdAM). The authors would like to thank and Elizabeth Mary Bevan for her careful linguistic revision of the first version.}}

\author{M. Aleandri\footnote{Scuola Normale Superiore, P.za dei Cavalieri 7, 56126, Pisa, Italy \textit{michele.aleandri@sns.it}. .
}\,  and M. Dall'Aglio\footnote{LUISS University, Viale Romania 32, 00197 Rome, Italy. \textit{mdallaglio@luiss.it}.}}   

\date{July 19, 2024 }

\begin{document}

	\maketitle
	\begin{abstract}
		We define a notion of the criticality of a player for simple monotone games based on cooperation with other players, either to form a winning coalition or to break a winning one, with an essential role for all the players involved. We compare it with the notion of differential criticality given by Beisbart that measures power as the opportunity left by other players.\\
		We prove that our proposal satisfies an extension of the strong monotonicity introduced by Young, assigns no power to null players and does not reward free riders, and can easily be computed from the minimal winning and blocking coalitions. An application to the Italian elections is presented. \\ Our analysis shows that the measures of group criticality defined so far cannot weigh essential players while only remaining an opportunity measure. We propose a group opportunity test to reconcile the two views. \\
		
	\end{abstract}
	\thispagestyle{empty}
	\section{Introduction}
	The commonly accepted distinction in game theory between non-cooperative and cooperative games would lead one to think that the former always examines the actions of individual players against the other players, while the latter focuses exclusively on groups of players, or coalitions, acting together. The reality, however, is more subtle. The notion of Nash equilibrium, a pillar of non-cooperative games, has been strengthened to include joint deviations by groups in the strong Nash equilibrium (see Aumann \cite{a59}). Conversely, many research works in cooperative game theory devoted great deal of effort into  defining indices that measure the importance of single agents as averages of their marginal relevance with respect to the other players (as in the Shapley value \cite{s51}) as if players were isolated in their action. 
	
	Focusing on cooperative games, interactions among groups of players take place in many phases: before playing the game in the process of coalition formation, and, once the coalitions are formed, interactions may continue among competing coalitions that define a partition of the players, and among players in the same coalition. We refer, for instance, to games in partition function form and to Koczy \cite{k18} for a recent review on the topic. 
	
	A less numerous body of work increased the number of players involved in the marginal increment analysis needed to measure player relevance. For simple monotone games, recent contributions have extended the notion of criticality to include players that may change an outcome of a game only through the help of other players, originating criticalities of higher order or rank. 
	The first proposal, the rank of differential criticality ($d$-criticality for short), was given by Beisbart \cite{b10}. Central to this definition is keeping criticality as a measure of the opportunity left by other players that are still valid for the higher ranks. More recently, Dall'Aglio et al. \cite{dfm16} introduced the notion of the order of criticality of a player to characterize situations where players may cooperate to break a winning coalition, with an essential role for all the players involved.
 
	In this work, the notion of the order of criticality provided by the latter proposal is extended to the losing coalitions to define the group essential criticality or $g$-criticality rank of a player. This creates a common setting for the comparison of the two notions of criticality. It turns out that both notions share a common core – that of the essential minimal critical (or $m$-critical) players. In turn, $m$-criticality determines the rank of $d$-criticality of the players that are not $m$-critical, so that the differences in the average rank of $d$-criticality between players are exclusively determined by their probability of not being minimally critical. Null players are never $g$-critical (and $m$-critical), while they are always $d$-critical in some rank. Moreover, $g$-criticality is the only notion that satisfies an extension of the strong monotonicity as introduced by Young, \cite{y85}, measured by the first order of stochastic dominance between vectors of average probability for all ranks. On the other hand, $d$-criticality is the only notion that measures criticality as the opportunity left by the other players. We question whether this opportunity test is the valid benchmark for comparing the action of groups and we close by proposing a notion of group opportunity that is satisfied by $g$-criticality.

	The paper is structured as follows: Section 2 recalls the definitions and results of previous works. In Section 3, the group essential criticality is introduced. In Section 4, some procedures for computing the various notions of group criticality are provided.
	In Section 5, several features of the competing definitions such as their monotonicity and their sensitivity to the null players are compared. In Section 6  two $g$-critical indexes are introduced and evaluated for the Italian election results of 2013, 2018 and 2022. Section 7 is devoted to a new look at the notion of opportunity. Section 8 concludes. 
	
	\section{Preliminary notions and available results}
	A \emph{simple cooperative game with transferable utility } (TU-game) is a pair $(N,v)$, where $N=\{1,2,\ldots,n\}$ denotes the finite set of players and $v:2^n\to\{0,1\}$ is the \emph{characteristic function}, with $v(\varnothing)=0$, $v(S)\leq v(T)$ for all $S,T$ subsets of  $N$ such that $S\subseteq T$ and $v(N)=1$.\\
	Given a coalition $S \subseteq N$, if $v(S)=0$ then $S$ is a \emph{losing} coalition, while if $v(S)=1$, then $S$ is a \emph{winning} coalition.  
	The marginal contribution of player $i$ to $S$ is defined as:
	\[
	v_i'(S)=\begin{cases}
		v(S)-v(S \setminus \{i\}) & \mbox{if } i \in S;
		\\
		v(S \cup \{i\} )-v(S ) & \mbox{if } i \notin S.
	\end{cases}
	\]
	The function $v_i'$ is called the \emph{derivative} of $v$ with respect to (wrt hereafter) $i$. A player $i$ such that $v_i'(S)=0$ for all $S\subseteq N$ is called \emph{null player}.
	
	We let $\wc=\{S \subseteq N: v(S)=1\}$ be the set of winning coalitions and $\wcm=\{W\in\wc: \nexists S\in\wc, S\subset W\}$ be the set of minimal winning coalitions. Moreover, let $\mathcal{B}=\{S\subseteq N: v(N)-v(N\setminus S)=1\}$ be the set of blocking coalitions and $\mathcal{B}_{\min}=\{B\in\mathcal{B}: \nexists S\in\mathcal{B}, S\subset B\}$ be the set of minimal blocking coalitions.
	
	Beisbart \cite{b10} proposes measures that ``quantify the extent to which a voter can make a difference as a member of a group".
	
	\begin{definition}(Definitions 3.1, 4.1, 4.2, A.1 and A.2 in \cite{b10})\label{def:dcritical}\noindent\\
		A coalition $G \subseteq N$ is \emph{critical} wrt a coalition $S$, if $S \cup G \in \mathcal{W}$ and $S \setminus G \notin \mathcal{W}$. If $G$ is critical wrt to $S$, $G$ is called \emph{critical inside (outside, resp.)} $S$, if $S \in \mathcal{W}$ ($S \notin \mathcal{W}$, resp.). A  player $i\in N$ is \emph{essential} for $G$ being critical wrt $S$ if $G\setminus\{i\}$ is not critical wrt $S$.\\
		A player  $i\in N$ is \emph{$d$-critical} of integer rank $\kappa$ wrt $S \subseteq N$ if there is $G \subseteq N$, with $i \in G$ and $|G|=\kappa$,  such that $G$ is critical wrt $S$ and $G$ has minimal cardinality, namely no other coalition $G'$ with $|G'|<\kappa$ and $i \in G'$ is critical wrt $S$.\\
        A player $i \in N$ is \emph{$e$-critical} of integer rank $\kappa$ wrt $S$, if there is $G \subseteq N$, with $i \in G$ and $|G|=\kappa$,  such that $G$ is critical wrt $S$ and $i$ is essential.
	\end{definition}

	For any coalition $S$, every player $i$ is $d$-critical of some rank $\kappa_i$ wrt S and we may have $\kappa_i\neq \kappa_j$ for pairs $i$ and $j$ of players. Moreover a player may fail to be $e$-critical of any rank.\\

Based on the above definitions, the author defines a voting power index as the probability for a player of being critical of a given rank with respect to a random coalition generated according to a probability distribution called {\em voting profile} (see Definition 4.3 in \cite{b10} for details). This extends the classical measurement of power based on the player’s solitary effort. Beisbart points out that the proposed indices should not depend on the player’s action, but only on the opportunity that the other players offer to
    the observed player. The notion of $d$-criticality is the only one satisfying this requirement and Beisbart's work focuses on $d$-criticality's properties, confining $e$-criticality to a bare definition in the work's appendix. Moreover, the given definition allows for many values of the rank, i.e.\ it defines a correspondence, since the rank of $e$-criticality for player $i$ wrt to a coalition $S$ depends on the chosen group of cooperating players in $G$.\\
    
	More recently, Dall'Aglio et al. \cite{dfm16} and Aleandri et al. \cite{adfm21} gave another definition of criticality that involves several players. This notion only considers the criticality inside $S$, i.e.\ when the player acts with others to make the coalition $S$ lose. A similar definition can be given for the case where the player acts to turn the coalition $S$ into a winner.
	\begin{definition}(Definition 2 in \cite{dfm16})\label{def:orderneg} Let $\kappa\geq 1$ be an integer and let $S\subseteq N$.\noindent\\
       Suppose $S$ is a winning coalition with $|S|\geq \kappa$. We say that a player $i$ is \footnote{In the original work, \cite{dfm16}, the player was simply referred to  as critical. We add here the term ``inside" to distinguish it from the complementary situation described in Definition \ref{def:orderneg}.2}\emph{inside critical of order} $\kappa$ wrt a  coalition $S$, and write $\rho^-=\kappa$, if $\kappa-1$ is the minimum integer such that there is a coalition $K\subseteq S\setminus\{i\}$ of cardinality $\kappa-1$ with 
		\begin{equation}\label{critical_gap}
			v\big(S\setminus K\big)-v\big(S\setminus(K\cup\{i\})\big)=1.
		\end{equation}
        Suppose $S$ is a losing coalition with $|S|\leq n-\kappa$. We say that a player $i$ is \emph{outside critical of order} $\kappa$ wrt a  coalition $S$, and write $\rho^+_i=\kappa$, if $\kappa-1$ is the minimum integer such that there is a coalition $K\subseteq N\setminus \left( S\cup \{i\}\right)$ of cardinality $\kappa-1$ with 
		\begin{equation} \label{critical_gap_pos}
			v\big(S\cup (K\cup\{i\})\big)-v\big(S\cup K\big)=1.
		\end{equation}
	\end{definition}

    The notion of inside criticality has been further investigated in the context of connection games in Dall'Aglio et al.\ \cite{dfm19a}, to define monotone indices of power (Dall'Aglio et al. \cite{dfm19b}) and to rank the players according to a lexicographic criterion (Aleandri et al.\cite{adfm21}). The notion of outside criticality has been seen as the dual counterpart of the inside criticality  in (Aleandri et al. \cite{afm22}) when the desirability binary relations between players is total, i.e. it affects every pair of players.

	\begin{example}\label{example1}
		Take $\mathcal{W}=\{ \{1,3\},\{1,2,3\}\}$ and $S=\{1\}$. Player 2 is $d$-critical and outside critical of rank/order $2$ via the coalition $G=\{2,3\}$, but it is not $e$-critical of any rank, because there is no way to make player 2 essential in changing the outcome. We observe that player $3$ is $d$-critical of rank $1$ via the smaller coalition $G'=\{3\}$. Moreover player $3$ is $e$-critical of rank $1$ through $G'$ and $e$-critical of rank $2$ through $G$. Finally, player 1 is $d$-critical of order $2$ via the coalition $G''=\{1,3\}$, but is neither inside critical nor e-critical of any order/rank. 
	\end{example}

	\section{Group criticality}
	In this section, we build upon the definitions given by Breisbart \cite{b10} and by Dall'Aglio et al.\cite{dfm16} to come up with a notion of group criticality that draws elements from both sources. We then make a comparison between the new proposal and that of $d$-criticality introduced by Beisbart.
	
	The non-essential players in a critical coalition have no effective power and it is natural to restrict our attention to essential players only. Moreover, Beisbart's proposal of $e$-criticality does not require all players in $G$ to be essential. Instead, we make this feature crucial in the following definition.
 
	\begin{definition}\label{def:essential critical}
		A coalition $G\subseteq N$ is called \emph{essential critical}, or simply \emph{essential},  wrt a coalition $S\subseteq N$, if each agent $i\in G$ is essential for $G$ being critical wrt $S$. Call $\mathcal{G}^e(S)$ the set of all essential coalitions wrt $S$.\\
		A player  $i\in N$ is \emph{group essential critical}, or simply \emph{$g$-critical}, of integer rank $\kappa$ wrt $S \subseteq N$ if there is  $G\in \mathcal{G}^e(S)$ of cardinality $|G|= \kappa$, containing player $i$, such that no other coalition $G'\in \mathcal{G}^e(S)$ exists with $i \in G'$ and $|G'|<\kappa$.
	\end{definition}
	
     For $S$ losing, every $G$ essential wrt $S$ has $G \cap S = \emptyset$ and so all the group essential players wrt $S$ belong to $N \setminus S$; for $S$ winning, every $G$ essential wrt $S$ has $G \cap S = G$ and so all the group essential players wrt $S$ belong to $S$.

The essential coalitions form a generating set for the set of all critical coalitions for a given $S$.
	\begin{proposition}
		Consider $S\subseteq N$ and $G$ a critical coalition wrt $S$. If we remove the non-essential players from $G$, we still have a critical coalition. Moreover, any coalition $G'$ such that $G\subseteq G'$ is critical wrt $S$.
	\end{proposition}
	\begin{proof}
		Let us write $G=E\cup nE$, where $E$ is the set of essential players and $nE$ is the set of non-essential players. According to the definition of a non-essential player, for all $i\in nE$, the coalition $G\setminus\{i\}$ is critical wrt $S$. The first part of the proposition is proved. For the second part it is sufficient to observe that $S\cup G'\supseteq S\cup G\in\wc$ and $S\setminus G'\subseteq S\setminus G\notin \wc $.
	\end{proof}
	
	Clearly, essential coalitions may have different cardinalities. 	We are now able to create the link between Beisbart's approach \cite{b10} and that of Dall'Aglio et al. \cite{dfm16}.
	\begin{proposition}
        Take $S\subseteq N$, then the following statements are equivalent:
        \begin{enumerate}[a)]
            \item player $i$ is $g$-critical of rank $\kappa$ wrt $S$;
            \item player $i$ is $e$-critical of minimal rank, i.e.\ its rank wrt $S$ is $\min E_i(S)$,  where 
            $$E_i(S)=\{ k\in\mathbb{N}:\mbox{ player i is $e$-critical of rank k wrt S} \}  ;$$
            \item player $i$ is either outside or inside critical of order $\kappa$ wrt $S$.
        \end{enumerate}
	\end{proposition}
	\begin{proof} Fix a coalition $S\subseteq N$.  We first observe that if player $i\in N$ is not $g$-critical of any rank then it is not essential. This implies that $v(S\setminus T)-v(S\setminus(T\cup\{i\})=0$ for all $T\subseteq S\setminus\{i\}$, and player $i$ is not inside or outside critical of any order. The converse is straightforward.\\
        $a)\Leftrightarrow b)$  Suppose that player $i$ is $e$-critical with minimal rank $\kappa$ by $G$, then there is no critical coalition $G'\subset N$, with $i\in G'$ and $|G'|<\kappa$. By minimality, any player $j\in G$ is essential to $G$ to be critical then player $i$ is $g$-critical with rank $\kappa$. \\
		  $a) \Rightarrow c)$ Let player  $i$ be $g$-critical of rank $\kappa$.  If $S \notin \wc$ then $i \notin S$ and there is an essential critical coalition $G\subseteq S^c$ (wrt $S$) with $i\in G$  such that  $|G|=\kappa$. Define $K=G\setminus\{i\}$, then equation \eqref{critical_gap_pos} holds and $\rho^+\leq \kappa$. If $\rho^+< \kappa $, by the minimality, there is an essential coalition $K'\subset S^c$ with $|K'|<|K|$ such that equation \eqref{critical_gap_pos} is satisfied implying that player $i$ is $g$-critical of rank $\kappa'+1<\kappa$. If $S \in \wc$ then $i \in S$ and the proof follows the same argument.\\
        $c)\Rightarrow a)$ Suppose that player $i$ is inside critical of order $\kappa$, then $i\in S$ and there is a coalition $K\subseteq S\setminus\{i\}$, $|K|=\kappa -1$, such that equation \eqref{critical_gap} is satisfied. By the minimality of $K$, coalition $G=K\cup\{i\}$  is essential critical wrt $S$, then player $i$ is $g$-critical of rank $\kappa$. If  player $i$ is outside critical of order $\kappa$ the proof is analogous and left to the reader.
        \end{proof}

    We conclude this section focusing on a particular subset of essential coalitions for $S\subseteq N$.
    
	\begin{definition}\label{def:mcritical}
		A coalition $G$ is \emph{essential minimal} for $S$ if $G=\arg\min\{|G'|: G'\in\mathcal{G}^e(S)\}$. Denote the family of all minimal essential coalitions for $S$ as $\mathcal{G}^e_{m}(S)$. A player $i$ is \emph{essential minimal critical}, or simply \emph{$m$-critical}, of rank $\kappa$ if there is $G\in\mathcal{G}_m^e(S)$ such that $i\in G$ and $|G|=\kappa$. If a player does not belong to any minimal essential coalition wrt $S$, it is called \emph{non-minimal}.
	\end{definition}
   Each player belonging to the same essential minimal coalition has the same \emph{$m$-critical} rank. On the other hand two players belonging to an essential coalition may have two different $g$-critical ranks and moreover a $g$-critical player may fail to be $m$-critical.

      \begin{example}
        Consider $\wc=\{\{1,2,3\},\{1,2,4,5\}\}$ and take $S=\{1\}$. Coalition $G=\{2,4,5\}$ is essential, but while  players $4$ and $5$ are $g$-critical of rank 3, player $2$ is $g$-critical of rank $2$ through coalition $G'=\{2,3\}\in\mathcal{G}^e_{m}(S)$. Players $4$ and $5$ are not $m$-critical of any rank and players $2$ and $3$ are $m$-critical of rank $2$.
      \end{example}
   
	\section{Computing group criticality}\label{sec:compute}
	We now focus on some computational techniques to derive the rank of criticality of a player according to the old and new definitions that we have reviewed so far. We begin with the newly introduced notion  and show a straightforward method to derive the rank of $g$-criticality wrt to any coalition from the collection of minimal winning coalitions $\mathcal{W}_{\min}$ or the dual notions of minimal blocking ones. A similar procedure had already been used to compute the order of inside criticality in Aleandri et al.\ \cite{adfm21}.  Given $\mathcal{C} \subseteq 2^N$ and $S \subseteq N$, we define $\mathcal{C} \setminus S =\left\{ C \setminus S:C \in \mathcal{C} \right\}$, and, for any player $i$, $\left(\mathcal{C}\right)_i= \left\{ C: i \in C \in \mathcal{C}  \right\} $. As acknowledged by Beisbart \cite{b10} (p.478), a similar straightforward relationship between $d$-criticality and minimal winning coalitions (or minimal blocking ones $\mathcal{B}_{\min}$) cannot be found\footnote{Beisbart actually mentions the collection of maximal losing coalitions, which are the complements, set by set, of the minimal blocking coalitions.}.

	\begin{proposition}\label{prop:calculation-g}
		Take $S \notin \mathcal{W}$, then player $i$ is $g$-critical if $\left(\mathcal{W}_{\min} \setminus S  \right)_i$ is non-empty with the rank given by the minimal cardinality of the sets in $\left(\mathcal{W}_{\min} \setminus S  \right)_i$. Taking $S \in \mathcal{W}$, player $i$ is $g$-critical if $\left(\mathcal{B}_{\min} \setminus S^c  \right)_i$ is non-empty with the rank given by the minimal cardinality of the sets in $\left(\mathcal{B}_{\min} \setminus S^c  \right)_i$.
	\end{proposition}
	\begin{proof}
		Take $S \notin \mathcal{W}$. Using hypothesis $T\in\left(\mathcal{W}_{\min} \setminus S  \right)_i$ with minimal cardinality, we show that $T$ is an essential critical coalition wrt $S$. By construction, $T$ is critical wrt $S$. Because of the criticality of $T$ there is $S_T\subseteq S$, such that $T\cup S_T\in\mathcal{W}_{\min}$. Now suppose that $T$ is not essential, then there is a $j\in T$ such that $T\setminus\{j\}$ is critical wrt $S$. This implies that $(T\setminus\{j\})\cup S_T\cup(S\setminus S_T)\in\mathcal{W}$. We know that $(T\setminus\{j\})\cup S_T\notin\mathcal{W}_{\min}$, then there is $A\subseteq (S\setminus S_T)$ such that $(T\setminus\{j\})\cup S_T\cup A\in\mathcal{W}_{\min}$, but then we obtain $T\setminus\{j\} \in \left(\mathcal{W}_{\min} \setminus S  \right)_i$ that contradicts the minimality of $T$.\\
		An analogous argument proves the second statement of the proposition.
	\end{proof}

	The following result clarifies the importance of $m$-critical players in defining the ranks of $d$-criticality. Players that are $m$-critical of a certain rank are $d$-critical and $g$-critical of the same rank. All the other players become $d$-critical of the rank immediately above, while there is no such direct link with $g$-criticality. Therefore, $d$-criticality is completely determined by $m$-criticality. 
	
	\begin{proposition}\label{prop:m+1=d}
		Given a coalition $S \subseteq N$,
		\begin{enumerate}[i.]
			\item if player $i \in N$ is $m$-critical wrt $S$ with rank $\kappa$, then $i$ is $d$-critical and $g$-critical of the same rank;
			\item otherwise, player $i$ is $d$-critical of rank $\kappa^m(S)+1$, where $\kappa^m(S)$ is the minimal size of the essential minimal critical coalitions for $S$, in formula
            \begin{equation}
		          \label{def:mS}
		          \kappa^m(S)= \min\{|G'|: G'\in\mathcal{G}^e(S)\}.
	            \end{equation}
		\end{enumerate}
	\end{proposition}
	\begin{proof}
		Denote as $\kappa_i^d(S)$ the rank for $d$-criticality of $i \in N$ wrt $S$ and suppose $i$ is $m$-critical with rank $k^m(S)$. Then the minimal essential coalition $\tilde{G}$ to which $i$ belongs is critical, therefore $\kappa_i^d(S) \leq k^m(S)$. The inequality cannot be strict, otherwise we would be able to extract  a minimal essential coalition  with a smaller cardinality than $k^m(S)$ from the critical coalition that defines the rank of $d$-criticality. Since $\tilde{G}$ is essential, a similar argument shows that player $i$ is also $g$-critical of the same rank.
		
		For any other player $j \in N$, consider the coalition $\hat{G}=G_m \cup\{i\}$ with $G_m \in \mathcal{G}^e_m(S)$. Now $\hat{G}$ is critical for $S$, with $|\hat{G}|=k^m(S)+1$. No smaller coalition containing $j$ does the same, otherwise the coalition would be minimal and we would fall under the previous case. Player $j$ is therefore $d$-critical of rank $k^m(S)+1$.
	\end{proof}
	
	The previous result shows that some players become $d$-critical without being essential for the corresponding critical coalition, gaining a rank that does not reflect the real impact on coalition formation.\\
	
	A player $i\in N$ is a \emph{free rider} for $d$-criticality if either $i$ is not $g$-critical of any rank or its rank for $g$-critical is higher than that of $d$-criticality.

	\begin{example}\label{ex:mdgranks}
		Consider 8 players and $\mathcal{W}_{\min}=\{\{1,2,3\},\{3,4,5\},$ $ \{4,5,6,7\}\}$ then the set of blocking coalitions is $\mathcal{B}_{\min}=\{\{1,4\},\{1,5\},\{2,4\},\{2,5\},\{3,6\},$ $\{3,7\}\}$. Take the losing coalition $S=\{1\}$ then $\mathcal{W}_{\min} \setminus S = \{\{2,3\},\{3,4,5\},$  $\{4,5,6,7\}\}$. By proposition \ref{prop:calculation-g} players 2 and 3 are $g$-critical of rank 2, players 4 and 5 are $g$-critical of rank 3 and players 6 and 7 are $g$-critical of rank 4. Player 8, being null and inessential for any coalition, is not $g$-critical of any rank.
  
        Here $\{2,3\}$ is the only minimal essential coalition so 2 and 3 are also $m$-critical and $d$-critical of rank 2. Players 4,5 are $d$-critical of rank 3. Finally, players 6, 7 and 8 too are $d$-critical of rank 3, so they are free riders.
		
		Note player 4 may form an essential coalition with 3 and 5 or may join the minimal essential coalition $\{2,3\}$, but this will not lower their rank. A similar argument occurs for player 5.
  
        Take the winning coalition $S=\{3,4,5,6,7\}$ then $\mathcal{B}_{\min} \setminus S^c=\{\{4\},\{5\},\{3,6\},$ $\{3,7\}\}$. By proposition \ref{prop:calculation-g} players 4 and 5 are $g$-critical of rank 1, while players 3, 6 and 7 are $g$-critical of rank 2. Again, the null player 8 is not $g$-critical of nay rank.
	\end{example}
	
	The previous example should convince the reader that free riding occurs when the rank for $g$-criticality is at least two units above the minimal rank of criticality and that null players are always free riders, but they are not the only ones.

	\section{A comparison of the properties}
	We now test the different notions of criticality according to several features.
 
	\subsection{Monotonicity}
	We base our analysis on an extension to rankings of the notion of monotonicity introduced by Young, \cite{y85}, which relies on criticality of the first order.
	
	\begin{definition}
		Let $\mathcal{T}$ be a class of cooperative games $(N,v)$ and let $\succeq$ be a preorder on a finite dimensional Euclidean space. An allocation procedure, that maps the class of games $\mathcal{T}$ onto the Euclidean space, satisfies the \emph{strong monotonicity} on $\mathcal{T}$ wrt the preorder $\succeq$ if, for any two games $v$, $w$ and any player $i \in N$, $ w_i'(S) \geq v_i'(S)$ for all  $S$ implies $\phi_i(w) \succeq \phi_i(v)$.
	\end{definition}
	\noindent
	Let $\mathcal{P}^N$ be the set of probability distributions on the power set of $N$. Take $p\in\mathcal{P}^N$, with $p(S)$ denoting the probability of coalition $S \subseteq N$ forming and new voting powers are introduced.
	\begin{definition}
 \label{def:powerorder}
		The {\em essential minimal measure} of voting power of rank $\kappa$ for player $i \in N$, denoted $\beta^{m,\kappa}_{i}$ is the probability of $i$ of being $m$-critical of rank $\kappa$ wrt to a random coalition $S \subseteq N$ originating from probability $p$.\\
		By replacing $m$-criticality with $g$-criticality we obtain the {\em group essential measure} of voting power of rank $\kappa$, $\beta^{g,\kappa}_{i}$,  while by replacing it with $d$-criticality, we obtain the {\em differential measure} of criticality, $\beta^{d,\kappa}_{i}$, always of rank $\kappa$.
	\end{definition}
	
	In Beisbart \cite{b10} the index $\beta^{d,\kappa}_{i}$ was simply referred to as the measure of voting power of rank $\kappa$. We add the term "differential" to distinguish it from the other values defined in the present work. \\

	Given a simple game $(N,v)$, we define, for each player $i\in N$,  the sequence of indices 
	$$\boldsymbol{\beta}^x_i(v):=(\beta^{x,1}_{i}(v),\ldots,\beta^{x,n}_{i}(v) ),$$ where $x=m$, $x=d$ or $x=g$. For our purposes, we are interested in a global index that includes the whole sequence of ranks. We consider $\pi_i^{x,\kappa} $, the probability of being $x$-critical up to rank $\kappa$,
	\[
	\pi_i^{x,\kappa}(v) = \sum_{\ell=1}^{\kappa} \beta^{x,\ell}_{i}(v) \; .
	\]
	\begin{definition}
		Let $v,w$ be simple games. For a given player $i \in N$, the vector of average ranks $\boldsymbol{\beta}^x_i(w)$   dominates $\boldsymbol{\beta}^x_i(v)$ in \emph{first-order stochastic dominance}, and we write $\boldsymbol{\beta}^x_i(w) \succeq_{\mathrm{1sd}} \boldsymbol{\beta}^x_i(v)$, whenever $  \pi_i^{x,\kappa}(w) \geq  \pi_i^{x,\kappa}(v) $ for every $ \kappa \in N$.
	\end{definition}
	

	We now examine Young's strong monotonicity for the different notions of criticality that we are comparing. We start with $g$-criticality.
	
	\begin{lemma}
		\label{lem:gmonotone}
		Given two characteristic functions $v$ and $w$ suppose that $ w_i'(S) \geq v_i'(S)$ for all $ S \subseteq N $ and some $ i \in N$. If $i$ is $g$-critical of rank $ \kappa$ wrt to a coalition in the game $v$, then $i$ is $g$-critical wrt to the same coalition in the game $w$ with a rank smaller than or equal to $ \kappa$.
	\end{lemma}
	\begin{proof}
		Denote as $G^T$ the essential coalition of minimal cardinality that defines the rank of $g$-criticality of $i$ wrt to $T$ in the game $v$. Since $w^i(S) \geq v^i(S) $, player $i$ remains essential for the criticality of $G^T$ wrt to $T$ under $w$, while other players may lose their essentiality.	The rank in $w$ may therefore remain the same or become lower.
	\end{proof}

	The strong monotonicity of $g$-criticality  follows.
	
	\begin{proposition}
		\label{prop:gmonotone}
		$g$-criticality satisfies strong monotonicity on the class of simple games  under the first-order stochastic criterion.
	\end{proposition}
	\begin{proof}
		Suppose $ w_i'(S) \geq v_i'(S)$. By Lemma \ref{lem:gmonotone}, if $i$ is $g$-critical under $v$, it is also $g$-critical under $w$ wrt to same coalition, with equal or smaller rank.
		Some other coalitions that failed to be $g$-critical under $v$ may become critical under $w$, but this can only increase the partial cumulative values of $\pi_i^{g,\kappa}$. The first-order stochastic dominance of the average ranks vector of $w$ wrt to $v$ follows.  
	\end{proof}
	
	The following examples illustrate the proposition at work and, most importantly, show that neither  $d$-criticality nor $m$-criticality share the same property.

	\begin{example}
		\label{ex:faild}
		Consider two simple games $v$ and $w$ on $N=\{1,2,3,4\} $ with minimal winning coalition families given, respectively, by $ \mathcal{W}_{\mathrm{min}}^v = \{\{1,2,3 \},\{4\}\}$ and $\mathcal{W}_{\mathrm{min}}^{w} = \{\{1,2,3 \},\{1,2,4\}\}$. It can easily be verified that $w_1'(S) \geq v_1'(S) $ for all $S \subseteq N $, with strict majorization for the coalitions $\{2,4\} $, $\{1,2,4\} $, $\{2,3,4\} $ and $N$.
		The presence of criticality of player 1, together with the rank, are described for the two games in Table \ref{table:mon}.
		
		\begin{table}[ht]
			\centering
			\begin{tabular}{||c|| c c|| c c ||} 
				\hline
				& \multicolumn{2}{|c|}{$d$-criticality} & \multicolumn{2}{|c|}{$g$-criticality}\\ \hline
				$S$ & $v$ & $w$ & $v$ & $w$ \\[0.5ex] 
				\hline\hline
				$\varnothing$ & 2 & \textcolor{red}{3} & 3 & 3 \\
				$\{1\}$ & 2 & \textcolor{red}{3} & \xmark & \xmark \\
				$\{2\}$ & 2 & 2 & 2 & 2 \\
				$\{3\}$ & 2 & 2 & 2 & 2 \\
				$\{4\}$ & 2 & 2 & \xmark & 2 \\
				$\{1,2\}$ & 2 & 2 & \xmark & \xmark \\
				$\{1,3\}$ & 2 & 2 & \xmark & \xmark \\ 
				$\{1,4\}$ & 2 & 2 & \xmark & \xmark \\
				$\{2,3\}$ & 1 & 1 & 1 & 1 \\
				$\{2,4\}$ & 2 & 1 & \xmark & 1 \\
				$\{3,4\}$ & 2 & 2 & \xmark & 2 \\
				$\{1,2,3\}$ & 1 & 1 & 1 & 1 \\
				$\{1,2,4\}$ & 2 & 1 & \xmark & 1 \\ 
				$\{1,3,4\}$ & 2 & 2 & \xmark & \xmark \\
				$\{2,3,4\}$ & 2 & 1 & \xmark & 1 \\
				$\{1,2,3,4\}$ & 2 & 1 & \xmark & 1 \\ 
				[1ex] 
				\hline
			\end{tabular}
			\caption{\label{table:mon} Comparing the ranks of group and differential criticality of player 1 for $v$ and $w$ in Example \ref{ex:faild}. Only the coalitions for which 1 is an essential critical player are listed.}
		\end{table}
		It turns out that $\beta^{g,1}_1(w)= \beta^{g,1}_1(v)+p(\{2,4\})+p(\{1,2,4\})+p(\{2,3,4\})+p(\{1,2,3,4\})$, $\beta^{g,2}_1(w)= \beta^{g,2}_1(v)+p(\{3,4\})+p(\{4\})$ and $\beta^{g,3}_1(w)= \beta^{g,3}_1(v)$. Therefore $\boldsymbol{\beta}_1^g(w) \succeq \boldsymbol{\beta}_1^g(v) $ for any $p \in \mathcal{P}^N $. Conversely, $\pi_1^{d,2}(v)=1 $ and $\pi_1^{d,2}(w)=1-p(\varnothing)-p(\{1\}) $. Therefore, $d$-criticality is not strongly monotone for every probability distribution $p \in \mathcal{P}^N$ such that either $p(\varnothing)>0$ or $p(\{1\})>0$ (or both).
	\end{example}

	\begin{example}
		\label{ex:failm}
		Now consider 5 players and the simple games $v$ and $w$ defined respectively by the minimal winning coalition sets  $ \mathcal{W}_{\mathrm{min}}^v = \{\{1,2,3,4 \}$, $\{2,3,4,5\}\}$ and $\mathcal{W}_{\mathrm{min}}^w = \{\{1,2,3 \},\{3,5\},\{2,4,5\}\}$. Again, $w_1'(S) \geq v_1'(S) $ for all $S \subseteq N $, with strict majorization for the coalitions $\{2,3\} $ and  $\{1,2,3\} $. The change in the ranks for both group and minimal criticality for player 1 are given  in Table \ref{table:mon2}.
		\begin{table}[ht]
			\centering
			\begin{tabular}{||c|| c c|| c c ||} 
				\hline \hline
				& \multicolumn{2}{|c||}{$g$-criticality} & \multicolumn{2}{|c||}{$m$-criticality}\\ \hline
				$S$ & $v$ & $w$ & $v$ & $w$ \\[0.5ex] 
				\hline\hline
				$\varnothing$ & 4 & 3 & 4 & \textcolor{red}{\xmark} \\
				$\{2\}$ & 3 & 2 & 3 & 2\\
				$\{3\}$ & 3 & 2 & 3 & \textcolor{red}{\xmark} \\
				$\{4\}$ & 3 & 3 & 3 & \textcolor{red}{\xmark} \\
				$\{2,3\}$ & 2 & 1 & 2 & 1 \\
				$\{2,4\}$ & 2 & 2 & 2 & \textcolor{red}{\xmark} \\
				$\{3,4\}$ & 2 & 2 & 2 &\textcolor{red}{\xmark} \\
				$\{1,2,3\}$ & \xmark & 1 & \xmark & 1 \\
				$\{2,3,4\}$ & 1 & 1 & 1 & 1 \\
				$\{1,2,3,4\}$ & 1 & 1 & 1 & 1 \\
				\hline \hline
			\end{tabular}
			\caption{\label{table:mon2} Comparing the ranks of group and minimal criticality of player 1 for $v$ and $w$ in Example \ref{ex:failm}. Only the coalitions for which 1 is an essential critical player are listed.}
		\end{table}
		
		All the coalitions which are group critical of some rank under $v$ remain critical under $w$, with the rank unaltered or decreased. Instead, some coalitions that where essential minimal critical under $v$, stop being critical under $w$. Any probability distribution assigning positive probability to those coalitions reveals that the notion of minimal criticality is non-monotone.
		
	\end{example}

	\begin{remark}
		In \cite{dfm19a}, a weighted sum of indices was considered
		and the notion of monotonicity introduced by Turnovec \cite{t98} for weighted voting games was proved for the inside order of criticality when all the coalitions are equally likely to occur. The assumption of equal probability was crucial for the result and  the use of weights, which might be considered arbitrary, make the present result more general.
		
		In \cite{adfm21}, another criterion is examined for single coalitions, the lexicographic order, defined as follows:  $\boldsymbol{\beta}^x_i(w) \succeq_{\mathrm{lex}} \boldsymbol{\beta}^x_i(v)$ if $\beta^{x,1}_{i}(w) > \beta^{x,1}_{i}(v) $ or, if $\beta^{x,1}_{i}(w) = \beta^{x,1}_{i}(v) $, when $\beta^{x,2}_{i}(w) > \beta^{x,2}_{i}(v) $. When both indices of rank 1 and 2 are equal, the third rank is compared, and so on. It is easy to verify that $\boldsymbol{\beta}^x_i(w) \succeq_{\mathrm{1sd}} \boldsymbol{\beta}^x_i(v)$ implies $\boldsymbol{\beta}^x_i(w) \succeq_{\mathrm{lex}} \boldsymbol{\beta}^x_i(v)$.
	\end{remark}	
	
	\subsection{Null players}
	Another problem with $d$-criticality comes from the fact that any player -- including null ones -- is critical of some rank. Therefore, power indices based on $d$-criticality will always assign non-null power to null players, as long as the index is based on a probability distribution that supports any coalition. This is not the case with $g$-criticality, since null players are always inessential.\\
	For a more rigorous analysis, we turn to the total probability of being critical according to the various criteria. 
	
	\begin{definition}
		We denote the probability of essential minimal (group essential, differential, resp.)   criticality for a player $i \in N$ as $\pi^x_i=\pi^{x,n}_i$  where $x=m$, $x=g$, $x=d$, resp. 
	\end{definition}
	
	\begin{proposition}\label{prop: relationPi}
		The following relationships hold for any $i \in N$,
		\begin{gather}
			\label{ineq_probcrit}
			\pi^m_i \leq \pi^g_i \leq \pi^d_i =1 \; .
		\end{gather}
	\end{proposition}
	\begin{proof}
		To prove \eqref{ineq_probcrit}, we simply note that minimal essential criticality is only a part of both essential and differential criticality, and therefore $\pi^m_i \leq \pi^g_i$ and $\pi^m_i \leq \pi^d_i$. The second inequality holds because not all agents are essential critical of any rank. 
	\end{proof}
	
	\begin{proposition}
		The following relations hold for any $i\in N$:
		\begin{gather}
			\mbox{ Player i is null } \Leftrightarrow \pi_i^g=0 \Leftrightarrow \pi_i^m=0.
		\end{gather}
	\end{proposition}
	\begin{proof} If player i is null it is easy to show that $\pi_i^g=0$, then by Proposition \ref{prop: relationPi} $\pi_i^m=0$.\\
		Suppose  that  player $i$ is not null.  Then they belong to some minimal winning coalition, $S$. We observe that the coalition $G=\{i\}$ is minimal essential wrt $S$ and then $\pi_i^m>0$, then $\pi_i^g>0$.
	\end{proof}

	\subsection{Average ranks}
	In section \ref{sec:compute}, we have shown that the rank of $d$-criticality is determined by a restricted group of $m$-critical players. This dependence is confirmed by the analysis of the average rank of differential criticality $\bar{\kappa}^d_i$, for a player $i \in N$, defined as
	\[
	\bar{\kappa}^d_i := \sum_{\kappa=1}^{n} \kappa \beta^{d,\kappa}_{i}.
	\]
	Since every player is $d$-critical of some order, we have
	\[
	\bar{\kappa}^d_i = \sum_{S \subseteq N}   \kappa_i^d(S) p(S),
	\]
	where $\kappa_i^d(S)$ is the $d$-critical rank of player $i$ wrt the coalition $S$. We now show that $\bar{\kappa}^d_i$ depends on two parameters of minimal criticality: the probability of being minimal essential critical and the  global average order of minimal criticality defined as
	\[
	\mu=\sum_{S \subseteq N} \kappa^m(S) p(S) \; .
	\]
	\begin{proposition} \label{prop:kid} For any player $i \in N$,
		\begin{equation}
			\label{eq:kid}
			\bar{\kappa}^d_i = 1+ \mu - \pi_i^m \; .
		\end{equation}
	\end{proposition}
	\begin{proof}
		For any $i \in N$, we have
		\begin{align*}
			\bar{\kappa}^d_i &= \sum_{S \subseteq N} k_i^d(S) p(S) =\sum_{\substack{S \subseteq N \\ i \: m \mbox{-critical}}} k_i^d(S)  p(S)+\sum_{\substack{S \subseteq N \\ i \mbox{ not } \: m \mbox{-critical}}} k_i^d(S)  p(S) 
			\\ &=
			\sum_{\substack{S \subseteq N \\ i \: m \mbox{-critical}}} k^m(S) p(S)+\sum_{\substack{S \subseteq N \\ i \mbox{ not } \: m \mbox{-critical}}} (k^m(S)+1) p(S) 
			\\ &=
			\sum_{S \subseteq N} k^m(S) p(S) + \sum_{\substack{S \subseteq N \\ i \mbox{ not } \: m \mbox{-critical}}} p(S) = \mu +1 - \pi_i^m. 
		\end{align*}
		
	\end{proof}
	The following corollary shows that the difference between the average rank of differential criticality of two players depends exclusively on the difference in their probability of essential minimal criticality. 
	\begin{corollary}
		Given two players $i,j\in N$, then
		\[ \bar{\kappa}^d_i-\bar{\kappa}^d_j=\pi_j^m-\pi_i^m.\]
	\end{corollary}

	\begin{example}
		Consider 4 players, $N=\{1,2,3,4\}$, and $\mathcal{W}_{\min}=\{\{1,2\},\{3\}\}$. Take a uniform probability distribution $p\in\mathcal{P}^N$, i.e.\ $p(S)=\frac{1}{16}$, $\forall S\subseteq N$. Then 
		
		\begin{equation*}
			\bar{\kappa}^d_1=\bar{\kappa}^d_2=\frac{7}{4}, \quad \bar{\kappa}^d_3=\frac{5}{4}, \quad \bar{\kappa}^d_4=\frac{17}{18} \; .
		\end{equation*}
		\vskip0.3cm 
		We next verify Proposition \ref{prop:kid}. We note that $k^m(S)=1$ for all the coalitions, but coalitions $\{1,2,3\}$ and $\{1,2,3,4\}$, for which $k^m(S)=2$. Therefore,
		\[
		\mu= \frac{9}{18} \; .
		\]
		Then we have
		\begin{align*}
			1+\mu-\pi_1^m & =  \frac{7}{4} = \bar{\kappa}^d_1 ;\quad
			1+\mu-\pi_2^m  =  \frac{7}{4} = \bar{\kappa}^d_2 ;\\
			1+\mu-\pi_3^m & =  \frac{5}{4} = \bar{\kappa}^d_3 ;\quad
			1+\mu-\pi_4^m  =  \frac{17}{8} = \bar{\kappa}^d_4 .
		\end{align*}
	\end{example}

 \section{An Electoral Application}
 Following the idea of Hausken and Mohr \cite{hm01} we analyse the $g$-criticality in the voting system of the Italian parliament with a  focus on the 2013, 2018, and 2022 elections.  \\
Italy has a bicameral legislative body and a winning coalition is a set of parties that have the majority in both Camera Dei Deputati and Senato. Up to 2018 the minimal age for passive and active electorate, were different in the two chambers and, moreover, the rules for allocating seats have changed between elections. Such changes make the changes of power indices from one election to the next of little significance. We focus instead on the comparison of the parties' indices and their distribution among the several ranks within the same election.\\

 We consider two indices, the first one is à la Shapley-Shubik and the second one is à la Banzhaf.  We want to show what is the power of a party according to its rank of $g$-criticality.
\begin{definition}
Let $N$ be a set of $n$ players. The {\bf Shapley-Shubik group essential power index} (or $g$-Shapley-Shubik value) of rank $\kappa$ for player $i$ is
\[
\phi^{g,\kappa}_{i} = \sum_{S \subseteq N} \frac{(n-s)!s!}{(n+1)!} I_g(i,\kappa,S),\]
where $s=|S|$ and 
\[I_g(i,\kappa,S) = \begin{cases}
1 & \small{\begin{array}{c} i  \mbox{ is $g$-critical of}  \\
 \mbox{rank }\kappa \mbox{ wrt }S \mbox{,} \end{array} }
\\
0  & \mbox{otherwise.}
\end{cases}
\]

\end{definition}
At first sight $\phi^{g,1}_{i}$ differs from the usual Shapley-Shubik index for player $i$ since the former counts both outside and inside criticality for any subcoalition of $N$, while the latter only records outside criticality for those coalitions of $N$ that do not include $i$. Indeed, the two indices coincide, since for each $S\subseteq N\setminus\{i\}$ such that $i$ is critical we have that $i$ is $g$-critical of order $1$ for both $S$ and $S\cup\{i\}$. Then we have to sum up the correspondent coefficients:
\begin{align*}
    \frac{(n-s)!s!}{(n+1)!} +  \frac{(n-s-1)!(s+1)!}{(n+1)!} &= \frac{(n-s-1)!s![(n-s)+(s+1)]}{(n+1)!} \\ & = \frac{(n-s-1)!s!}{n!},
\end{align*}
that is exactly the usual Shapley coefficient.
\begin{definition}
Let $N$ be a set of $n$ players. The {\bf Banzhaf group essential power index} (or $g$-Banzhaf value) of rank $\kappa$ for player $i$ is
\[
\beta^{g,\kappa}_{i} =\frac{1}{2^n}\sum_{S \subseteq N} I_g(i,\kappa,S).
\]
\end{definition}

By means of a reasoning similar to that already employed for the Shapley-Shubik index, it can be verified that $\beta^{g,1}_{i}$ is the usual Banzhaf value for player $i$.\\

\noindent
\textit{2013 ELECTIONS:}\\
We have 15 parties that have positive $g$-Shapley-Shubik and $g$-Banzhaf values. In Table \ref{table:2013} there are the results of the election. The government was form by the coalition Partito Democratico (PD), Scelta Civica (SC), Unione di Centro (UC), S\"udtiroler Volkspartai (SV), Centro Democratico (CD), Futuro e Liberta per l'Italia (FL), Movimento Associativo Italiani all'Estero (MAIE) and other independents. 
\begin{table}[h]
\centering
    \begin{tabular}{|m{5cm}||m{2.5cm}|m{2.5cm}|m{2cm}|}
    \hline
    \bf Party & \bf Acronym & \bf Camera dei deputati & \bf Senato \\
    \hline
        Partito Democratico & PD &  297 & 111 \\
        Il Popolo della Libert\'a & PdL & 98 & 98 \\
        Movimento 5 Stelle & M5S & 109 & 54 \\
        Lega Nord & LN & 18 & 18 \\
        Scelta Civica & SC & 38 & 15 \\
        Sinistra Ecologia e Libert\'a & SI & 37 & 7 \\
        Unione di Centro & UC & 8 & 2 \\
        S\"udtiroler Volkspartei & SV & 5 & 4 \\ 
        Fratelli d'Italia & FdI & 9 & 0 \\
        Centro Democratico & CD & 6 & 1  \\
        Futuro e Libertà per l'Italia & FL & 1 & 2 \\
        Movimento Associativo Italiani all'Estero & MAIE & 2 & 1 \\
        Vallée d'Aoste (UV-SA-FA) & VA & 1 & 1 \\
        Grande Sud & GS & 0 & 1 \\
       	Unione Sudamericana Emigrati Italiani & USE & 1 & 1\\
        \hline
    \end{tabular}
    \caption{Seats in the two chambers for the 2013 Elections}
    \label{table:2013}
\end{table}

\begin{center}
    \includegraphics[width=0.80\textwidth]{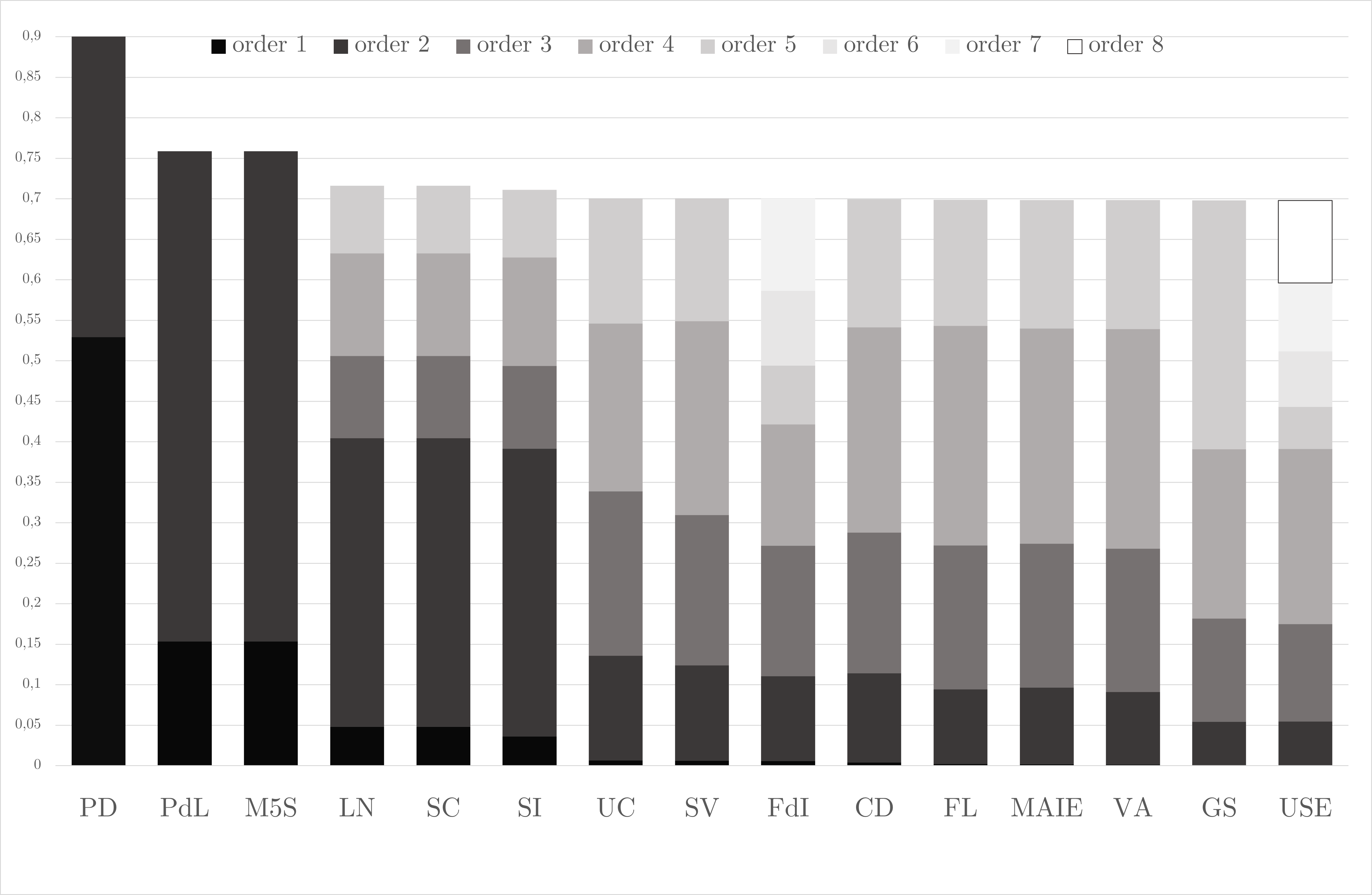} \\  
    {$g$-Shapley-Shubik values for the 2013 Election. }
\end{center}
\begin{center}
    \includegraphics[width=0.80\textwidth]{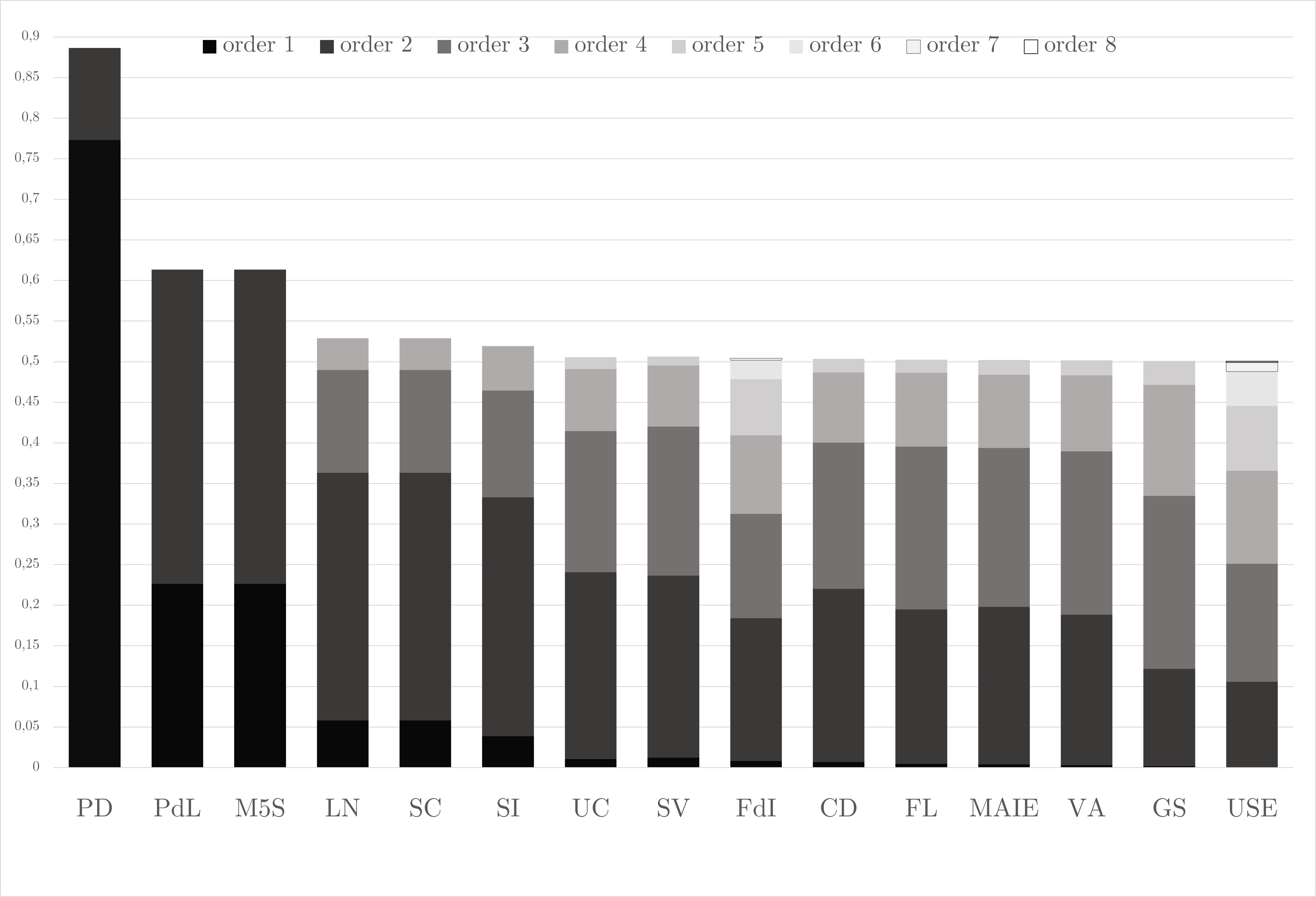}\\
    {$g$-Banzhaf values for the 2013 election. }
\end{center}
We note that PD has a large $g$-Shapley value but a small $g$-Banzhaf value of rank 2, this is due to the different probabilities assigned to the coalitions containing that party for the Shapley and Banzhaf values. The government  is formed by 5 parties with large $g$-Shapley value of rank 5 and large $g$-Banzhaf value of rank 3.\\

\noindent
\textit{2018 ELECTIONS:}\\
We have only 4 parties that have positive g-Shapley and g-Banzhaf values, the other 10 parties have no power. In Table \ref{table:2018} there are the results of the elections. There were three governments:
\begin{enumerate}
    \item Movimento 5 Stelle (M5S), Lega Nord (LN) + Others
    \item M5S, PD + Others
    \item M5S, LN, PD, Forza Italia (FI) + Others
\end{enumerate}
\begin{table}[h]
\centering
    \begin{tabular}{|m{5cm}||m{2.5cm}|m{2.5cm}|m{2cm}|}
    \hline
    \bf Party & \bf Acronym & \bf Camera dei deputati & \bf Senato \\
    \hline
        Movimento 5 Stelle & M5S & 227 & 112 \\
        Lega  & LN & 125 & 58 \\
        Partito Democratico & PD &  112 & 53 \\
        Forza Italia & FI & 104 & 57 \\
        Fratelli d'Italia & FdI & 32 & 18 \\
        Liberi e Uguali & LU & 14 & 4 \\
        Noi con l'Italia & UDC & 4 & 4 \\
        S\"udtiroler Volkspartei & SV & 4 & 3 \\ 
        +Europa & +Eu & 3 & 1 \\
        Civica popolare & CP & 2 & 1  \\
        Italia Europa Insieme & IEI & 1 & 1 \\
        Movimento Associativo Italiani all'Estero & MAIE & 1 & 1 \\
       	Unione Sudamericana Emigrati Italiani & USE & 1 & 1\\
        Union Vald\^otaine & UV & 0 & 1 \\
        \hline
    \end{tabular}
    \caption{Seats in the two chambers for the 2018 Elections}
    \label{table:2018}
\end{table}

\begin{center}
    \includegraphics[width=0.40\textwidth]{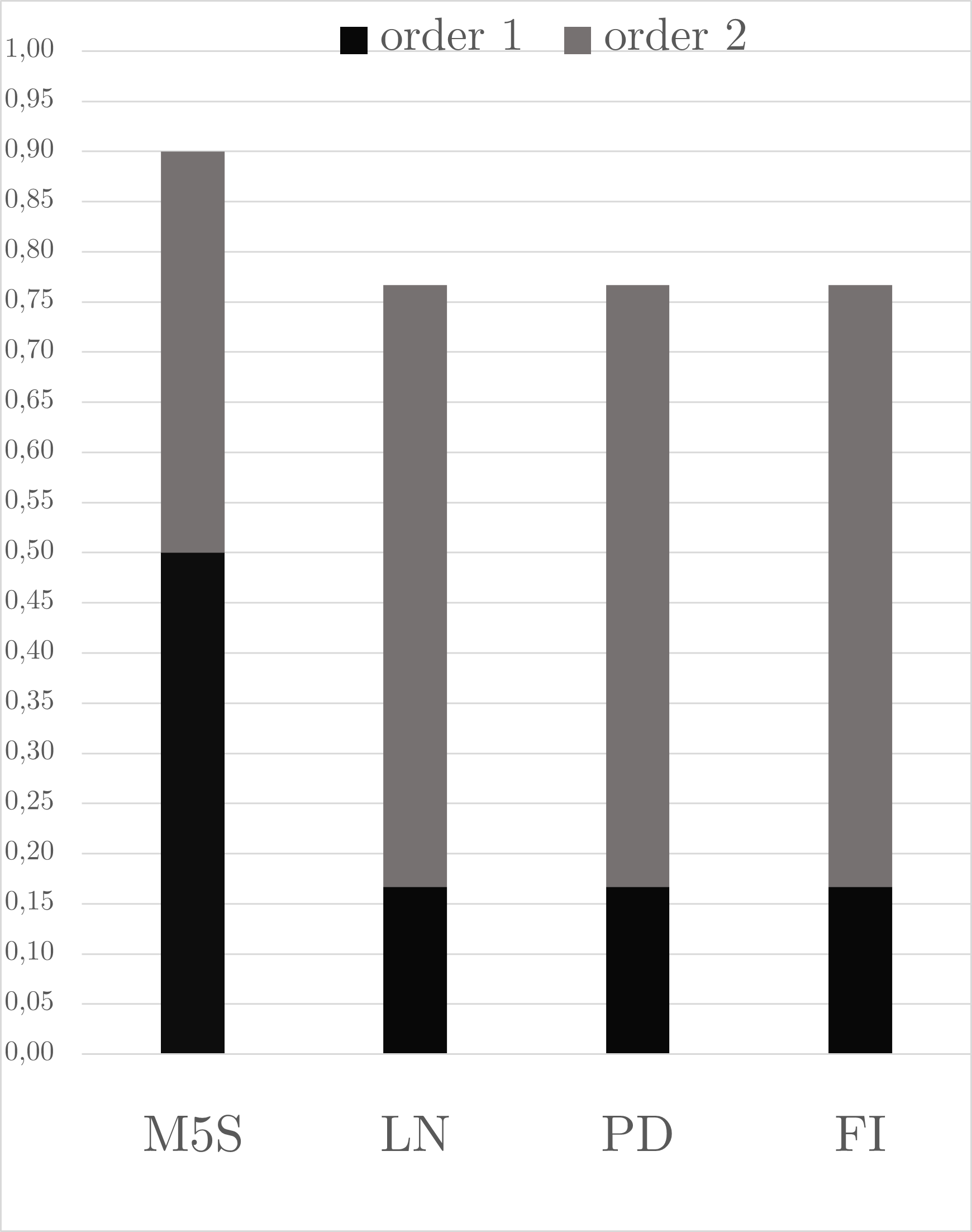}\quad
     \includegraphics[width=0.40\textwidth]{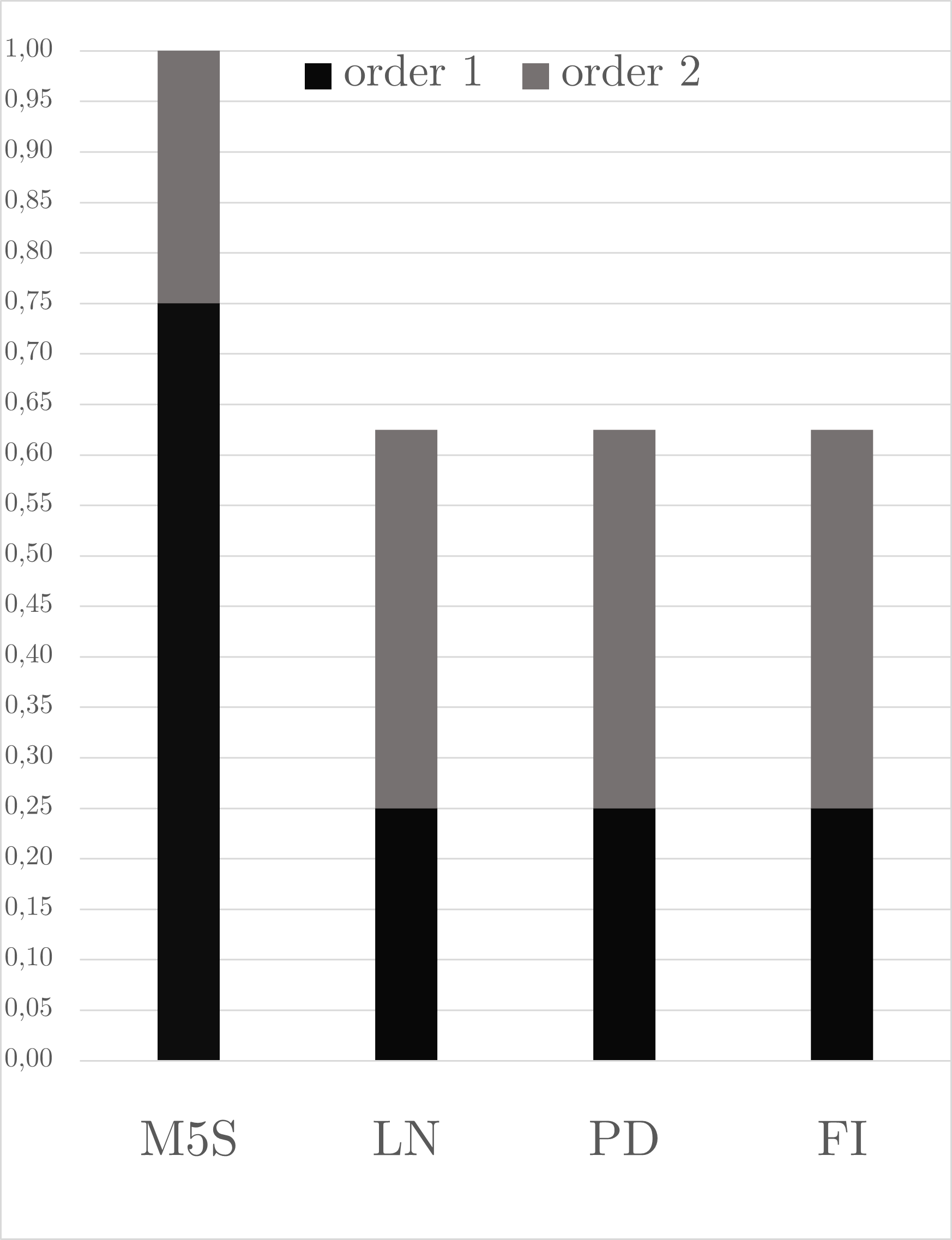}\\
    {The $g$-Shapley-Shubik (left) and $g$-Banzhaf (right) values for the 2018 elections. }
\end{center}
We observe that in this election the governments was formed by MS5 that has the highest value for the $g$-Shapley and $g$-Banzhaf values of rank 1. On the other hand the other parties have the same values for the g-Shapley and g-Banzhaf values of rank 2.\\

\noindent
\textit{2022 ELECTIONS:}\\
We have 14 parties that have positive $g$-Shapley and $g$-Banzhaf values. In Table \ref{table:2022} there are the results of the elections. The government is formed by Fratelli d'Italia (FdI), Lega, FI and Noi Moderati (NM).
\begin{table}[h]
\centering
    \begin{tabular}{|m{5cm}||m{2.5cm}|m{2.5cm}|m{2cm}|}
    \hline
    \bf Party & \bf Acronym & \bf Camera dei deputati & \bf Senato \\
    \hline
        Fratelli d'Italia & FdI & 119 & 65 \\
        Partito Democratico & PD &  69 & 40 \\
        Lega per Salvini Premier  & Lega & 66 & 30 \\
        Movimento 5 Stelle & M5S & 52 & 18 \\
        Forza Italia & FI & 45 & 28 \\
        Azione - Italia Viva & Az IV & 21 & 9 \\
        Alleanza Verdi e Sinistra & V e S & 12 & 4 \\
        Noi Moderati & NM  & 7 & 2 \\
        S\"udtiroler Volkspartei & SVP & 3 & 2 \\ 
        +Europa & +Eu & 2 & 0 \\
        Impegno Civico - Centro Democratico & IC & 1 & 0  \\
        Sud Chiama Nord & ScN & 1 & 1 \\
        Union Vald\^otaine & UV & 1 & 0 \\
        Movimento Associativo Italiani all'Estero & MAIE & 1 & 1 \\
        \hline
    \end{tabular}
    \caption{Seats in the two chambers for the  2022 Elections}
    \label{table:2022}
\end{table}
\begin{center}
    \includegraphics[width=1\textwidth]{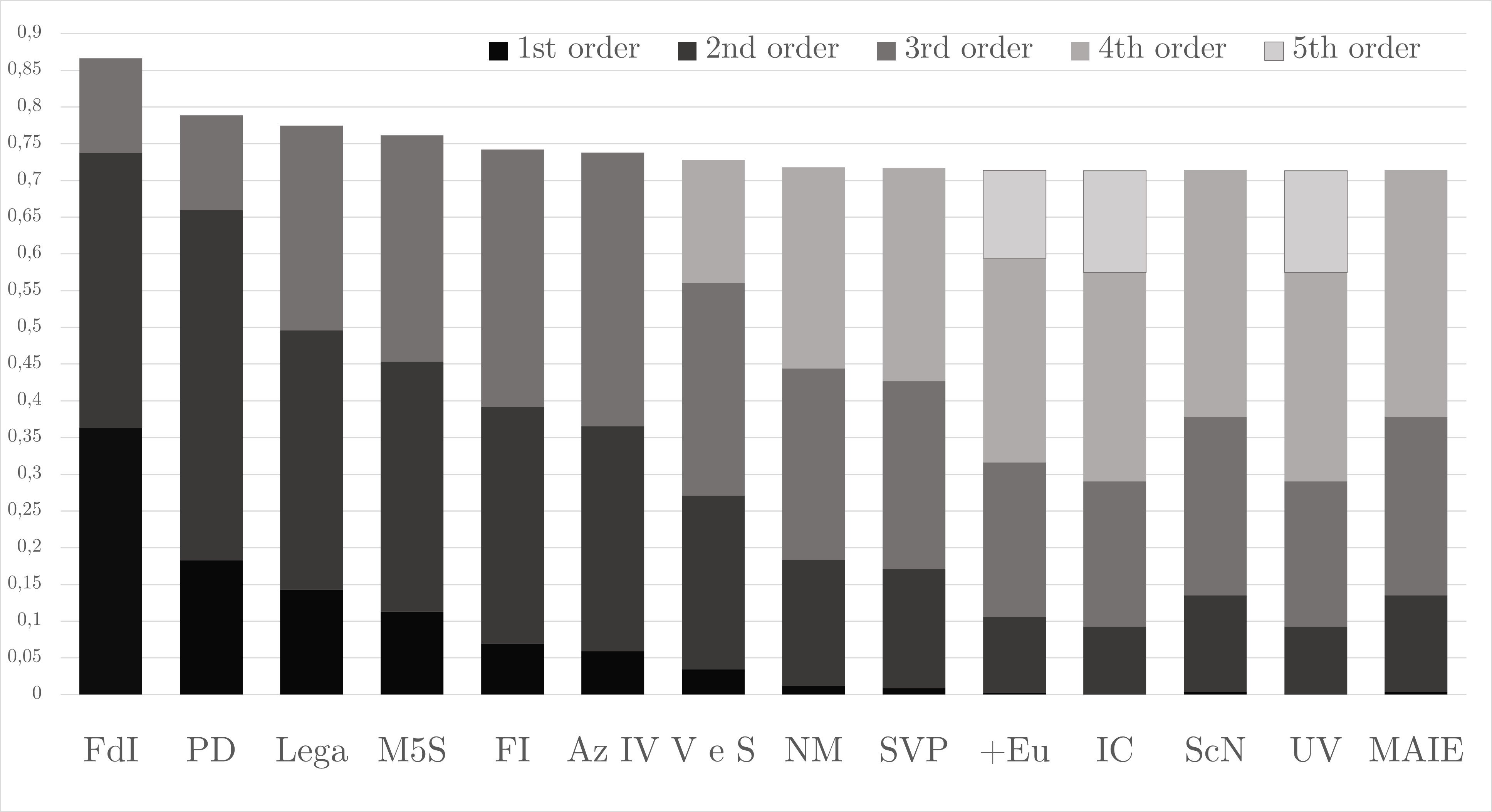}
    {$g$-Shapley-Shubik values for the 2022 election. }
\end{center}
In this election parties Az-IV and FI have the highest $g$-Shapley and $b$-Banzhaf values of rank 3, respectively.
\begin{center}
     \includegraphics[width=\textwidth]{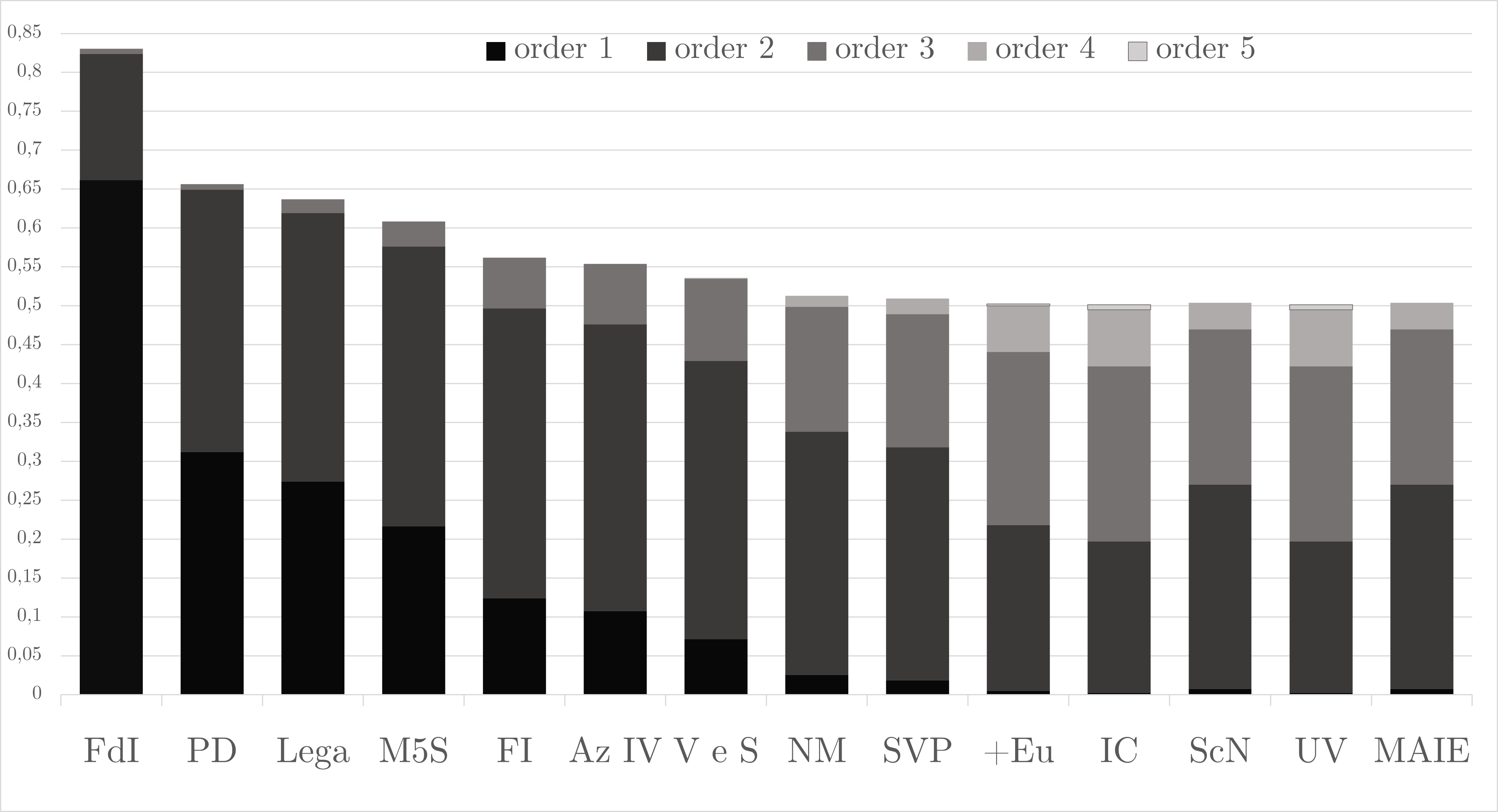}\\
    {$g$-Banzhaf values for the 2022 election. }
\end{center}


The electoral application gives us some takeaways: distributions in the parliament seats that are quite similar may lead to drastically different distributions in terms of power. In the 2013 and 2022 elections all parties, including those with only 1 seat had some higher rank power index, and the total probability of being critical of some rank, i.e.\ the sum of power indices of all ranks for a single party, is very similar for all the parties, from the most voted to those having a single representative. On the other hand, in the 2018 election, only 4 out of the 14 parties detain some power, and each of them needs at most one other party to achieve the majority. The remaining 10 parties have no power of sorts, according to the Shapley, the Banzhaf or any other index based on Definition \ref{def:powerorder}. Furthermore, the application shows a monotonic behaviour of the indices wrt to the number of seats: if party $A$ has more seats than party $B$ in both chambers, the vector of average ranks of party $A$ dominates  that of party $B$ in first-order stochastic dominance. Such dominance is easy to prove. Indeed we can apply proposition \ref{prop:gmonotone} taking a game where party $A$ has more seats than party $B$ in both chambers and another game where party $A$ has the same number of seats than party $B$ in both chambers. 

Some pair of parties cannot be compared because there is no clear winner in comparing the seats in the two chambers. When we consider the single average ranks of the players, instead, no monotonic behavior emerges.

\section{Opportunity Measures that Involve More Players}
Our attention now turns to the concept of a player's power, assessed as the opportunity afforded to them by the other players. This analysis is prompted by the role assigned to this feature by Beisbart, which guided his decisions in formulating a criticality measure for groups.
    In the author's words, opportunity will take place "if and only if the configuration of the other votes is such that, whether a bill passes or not depends on her vote" (\cite{b10}, p.472), with the votes of the other players taken as given. Consequently, any proposal to define the voting power of a player should satisfy the following benchmark.
	
	\begin{definition}
		\label{def:optest}
		A criticality notion passes the \emph{opportunity test}, if whenever a player $a \in N$ is critical wrt $S$, then, the same player is also critical wrt $S \cup \{a\}$ and $S \setminus \{a\}$. 
	\end{definition}
	
	Beisbart's notion of $d$-criticality passes the test 	for any rank, while $e$- criticality and $g$-criticality require
 a player to be indispensable in overturning the outcome and, therefore, fails the same test. Take for instance $S =\{1\}$ in Example \ref{ex:mdgranks}. While player 2 is $d$-critical of rank 2 for both $S \setminus\{2\}=S$ and $S \cup \{2\}=\{1,2\}$, the same player is $g$-critical of the same rank for $S$ alone, while it is inessential for $S \cup \{2\}$, and therefore not $g$-critical of any rank. On the other hand, as pointed out in the previous sections, $g$-criticality satisfies many properties that the differential counterpart fails to satisfy.

 At first sight there does not seem to exist an uncontested winner among the proposals for extending the notion of criticality to actions involving groups of players. A closer look makes us question the adequacy of the opportunity test given in Definition \ref{def:optest} for notions of criticality that involve groups of players. This indicator
 was defined in a context where voting power measures ``quantify the extent (denoted as $E_\kappa$) to which $i$ can be a member of a group of size $\kappa$ that has the opportunity to make a difference as to whether the bill passes" (Beisbart \cite{b10}, p.473). 

When the first rank is considered, all the notions of criticality considered in this work pass the test, since a player $i$ is outside critical for $S \setminus \{i\}$ if and only if it is inside critical for $S \cup \{i\}$.

 Moving to higher ranks and focusing on $d$-criticality, an outside critical player for $S \setminus \{i\}=S$ cannot become inside critical for $S \cup \{i\}$, because other players outside $S$ are needed to change the outcome. These players are jointly critical wrt to $S \cup \{i\}$, making $i$ inessential for the change. In fact, the following result regarding $d$-criticality holds.
	
	\begin{proposition}
		Suppose player $i \in N$ is $d$-critical of rank 2 or higher wrt $S$. Then $i$ is inessential for at least one of the critical coalitions that define the rank of $d$-criticality in $S \setminus \{i\}$ and $S \cup \{i\}$.
	\end{proposition}
	\begin{proof}
		If and $i \notin S \in \mathcal{W}$, then $i$ is always inessential for the critical coalition that turns $S=S\setminus\{i\}$ into a losing coalition. The case when $i \in S \notin \mathcal{W}$ and $i$ is not essential for $S=S \cup \{i\}$ is similar.
		
		Consider now $i \in S \in \mathcal{W}$. Since $i$ is $d$-critical with a rank greater than 1, then $S \setminus \{i\} \in \mathcal{W}$, with $i$ not essential for the critical coalition that turns $S\setminus\{i\}$ into a losing one. The case $i \notin S \notin \mathcal{W}$ is treated similarly.
	\end{proof}
 
When a group of players is involved, $d$-criticality satisfies the opportunity test, but the price to pay is to make inessential players critical in at least half of the cases where their $d$-criticality rank is 2 or higher. The fraction of cases may be higher than that, since a player may be inessential for both  $S\setminus\{i\}$ and  $S\cup\{i\}$. Consider for instance player 6 in Example \ref{ex:mdgranks}, who is inessential for both  $S\setminus\{6\}$ and $S\cup\{6\}$ when $S=\{1\}$.\\
	
Once the context changes from the situation where the action of single players is measured, to the context where several other players are involved, we believe that the notion of the opportunity left by the other players should change in such a way that it reflects the joint action of players. We therefore suggest the following proposal.

\begin{definition}
    A criticality notion for player $i$ wrt to a coalition $S$ satisfies {\em the opportunity test for an action involving a coalition  of other players $G \subset N \setminus \{i\}$} if $i$ is either critical for the pair $(S \cup G) \setminus \{i\}$ and $S \cup G \cup \{i\}$ or it is critical for the pair $(S \setminus G) \setminus \{i\}$ and $(S \setminus G) \cup \{i\}$.
\end{definition}
Clearly, if the action of player $i$ involves no other player, then the test coincides with the familiar Definition \ref{def:optest}.
It can be easily verified that $g$-criticality passes the new test, while guaranteeing the essential role of the players.

Differently from what happens with $d$-criticality, a player $i$ that is $g$-critical wrt $S$ of rank two or higher will not be simultaneously critical for $S \setminus \{i\}$ and $S \cup \{i\}$ and the symmetry displayed in the previous situation breaks down. If $i$ is outside $g$-critical of rank $\kappa \geq 2$ wrt $S$, it will not be $g$-critical for $S \cup \{i\}$. Instead it will be inside $g$-critical of rank 1 wrt $S \cup G$ where $G$ is some essential critical coalition for $S$ of size $\kappa - 1$. Incidentally, $i$ will also be outside $g$-critical wrt all the coalitions $S \cup G'$ with $i \notin G'$ and with the rank given by $|G \setminus G'|+1$.  Similarly, if $i$ is inside $g$-critical of rank $\kappa \geq 2$ wrt $S$, it will not be $g$-critical for $S \setminus \{i\}$. Instead it will be outside $g$-critical of rank 1 wrt $S \setminus G$ where $G$ is some essential critical coalition for $S$ of size $\kappa - 1$. Incidentally, $i$ will also be inside $g$-critical wrt all the coalitions $S \setminus G'$ with $i \notin G'$ and with the rank given by $|G \setminus G'|+1$.

	We note that symmetry holds when first order rank is considered in line with what is already known about classical criticality. Higher order criticality, instead, contributes to a single side alone: outside criticality when $S \notin \mathcal{W}$ and inside criticality otherwise. This asymmetric part arises from the collaborative effort among players.

	\section{Conclusions}
	By extending the notion of order of criticality given in Dall'Aglio et al.\ \cite{dfm16} for inside players, we have defined the rank of group essential criticality of a player and we have compared this notion with that of differential criticality from Beisbart \cite{b10}. For any given coalition, the two notions coincide for those players who belong to the essential critical coalitions of minimal cardinality, who are labelled minimal critical, or $m$-critical players, but usually differ for the remaining players. Group criticality satisfies several properties that differential criticality lacks:  it is strongly monotone in the sense defined by Young \cite{y85}, it is never associated with null players and it is not deterministically determined by the set of $m$-critical players. Conversely, $d$-criticality is the only criterion that is compatible with the principle that a measure of criticality should evaluate the degree of opportunity that the other players give. When several players work together to overturn a game's outcome, the two principles of opportunity and essentiality in their current definitions  are incompatible and a decision must be taken on which one to save. We believe that this hiatus  originates in an improper application of the opportunity principle in a context where players do not act alone, but work together to reach their goals. For this reason and with the goal of reconciling the two principles, we propose an opportunity test involving several players which is satisfied by the newly introduced notion of group criticality.
	
	Regarding future research directions, we believe that investigating the relationship between indices that measure the importance of the single players by evaluating their criticality jointly with others and those that measure the importance of coalitions as a whole should be analysed in more detail. We refer to some extensions of power indices to include whole coalitions instead of single players. Both the Shapley and the Banzhaf values have been extended to measure the interaction among players in Grabisch and Roubens \cite{gr98}. In a similar vein, Hausken and Mohr \cite{hm01} have defined the Shapley value of one player to another player and this work has been recently extended by Hausken \cite{h20} to define the Shapley value of one coalition to another coalition.

\end{document}